\documentclass[10pt]{article}
\usepackage{fullpage}
\usepackage[english]{babel}
\usepackage[cmex10]{amsmath}
\usepackage[all=normal,paragraphs=tight,bibliography=tight]{savetrees}
\usepackage{todonotes}
\usepackage{amsthm}
\usepackage{amssymb}
\usepackage{microtype}
\usepackage{xspace}
\usepackage{comment}
\usepackage{letltxmacro}
\usepackage{comment}
\usepackage{enumerate}
\usepackage{enumitem}
\usepackage{tikz}
\usetikzlibrary{calc}
\usepackage{url}

\usepackage{graphicx}
\usepackage{caption}
\usepackage{subcaption}
\usepackage{float}

\newtheorem{theorem}{Theorem}[section]
\newtheorem{lemma}[theorem]{Lemma}
\newtheorem{myclaim}[theorem]{Claim}

\theoremstyle{definition}
\newtheorem{definition}{Definition}[section]

\def\cqedsymbol{\ifmmode$\lrcorner$\else{\unskip\nobreak\hfil
\penalty50\hskip1em\null\nobreak\hfil$\lrcorner$
\parfillskip=0pt\finalhyphendemerits=0\endgraf}\fi} 

\newcommand{\cqed}{\renewcommand{\qed}{\cqedsymbol}}
\newcommand{\executeiffilenewer}[3]{%
\ifnum\pdfstrcmp{\pdffilemoddate{#1}}%
{\pdffilemoddate{#2}}>0%
{\immediate\write18{#3}}\fi%
} 
\newcommand{\dirmc}{\textsc{Directed Multicut}}

\newcommand{\subiso}{\textsc{Partitioned Subgraph Isomorphism}}
\newcommand{\terms}{\mathcal{T}}

\newcommand{\stor}{\textsc{Steiner Orientation}}
\newcommand{\Oh}{\mathcal{O}}

\title{Directed multicut is $W[1]$-hard, even for four terminal pairs\thanks{An extended abstract of this work has been presented at SODA 2016~\cite{soda}.}}
\author{Marcin Pilipczuk\thanks{University of Warsaw, Poland, \texttt{malcin@mimuw.edu.pl}. Research done while the author was at University of Warwick, partially supported by the Centre for Discrete Mathematics and its Applications (DIMAP) at the University of Warwick and by Warwick-QMUL Alliance in Advances in Discrete Mathematics and its Applications.} \and Magnus Wahlstr\"{o}m\thanks{Royal Holloway University of London, UK, \texttt{Magnus.Wahlstrom@rhul.ac.uk}}}

\date{}

\begin{document}
\pagenumbering{gobble}
\thispagestyle{empty}

\maketitle

\begin{abstract}
We prove that \textsc{Multicut} in directed graphs, parameterized by the size of the cutset, is $W[1]$-hard and hence unlikely
to be fixed-parameter tractable even if restricted to instances with only four terminal pairs. 
This negative result almost completely resolves
one of the central open problems in the area of parameterized complexity of graph separation problems,
posted originally by Marx and Razgon [SIAM J. Comput. 43(2):355--388 (2014)], leaving only the case of three terminal pairs open.
The case of two terminal pairs was shown to be FPT by Chitnis et al.~[SIAM J. Comput. 42(4):1674--1696 (2013)].

Our gadget methodology also allows us to prove $W[1]$-hardness of the \textsc{Steiner Orientation} problem
parameterized by the number of terminal pairs, resolving an open problem of Cygan, Kortsarz, and Nutov [SIAM J. Discrete Math. 27(3):1503-1513 (2013)].
\end{abstract}

%\clearpage
%\pagenumbering{arabic}

\section{Introduction}\label{sec:intro}
The study of cuts and flows is one of the central
areas of combinatorial optimization,
with decades of intensive research on polynomial-time exact and approximate algorithms.
%both for the basic min-cut problem and for more general graph separation problems.
Since the seminal work of Marx~\cite{marx-impsep}, 
it has also become one of the most dynamic research directions in
parameterized complexity, under the name of \emph{graph separation problems}.

The key contribution of Marx' paper~\cite{marx-impsep} lies in the notion
of an important separator, which is a kind of greedy cut that is useful in
some graph cut problems. Marx showed that the number of important
separators of size up to $k$ in a graph is bounded as a function of $k$
alone, and gave applications to a number of graph separation problems, in
particular an FPT algorithm for \textsc{Multiway cut}~\cite{marx-impsep}.\footnote{In the \textsc{Multiway Cut} problem, the input consists of a graph $G$, a set $T \subseteq V(G)$ of terminals, and an integer $p$; the goal is to delete at most $p$ edges (in the edge-deletion variant) or nonterminal vertices (in the node-deletion variant) so that every terminal lies in a different connected component
    of the resulting graph, that is, to separate all pairs of terminals.}
(An improved bound was given by Chen et al.~\cite{ChenLL09}, with a simple
but influential proof method.) These techniques also led to FPT algorithms
for \textsc{Directed Feedback Vertex Set}~\cite{dfvs-alg} and eventually
for \textsc{Almost 2-SAT}~\cite{a2sat-alg}, which is a powerful 
generic problem encapsulating a number of other well-studied problems, 
such as \textsc{Odd Cycle Transversal} 
and \textsc{Vertex Cover} parameterized by the excess above a maximum matching.
Further study of the graph separation problems in the realm of parameterized complexity
revealed a plethora of other algorithmic techniques:
the technique of shadow removal~\cite{marx-razgon} lead to an FPT algorithm
for the directed version of \textsc{Multiway Cut}~\cite{dmwc-alg},
the framework of randomized contractions~\cite{rand-contr}
lead to an FPT algorithm for \textsc{Minimum Bisection}~\cite{bisection-alg},
whereas the idea to guide a branching algorithm
by the solution to the LP relaxation~\cite{mwc-lp}
lead not only to the fastest known algorithm for \textsc{Odd Cycle Transversal}~\cite{vc-lp},
but also can be cast into a very general CSP framework~\cite{magnus-lp}.

In tandem with these developments, we have gained a growing understanding
of the structure of bounded-size cuts in graphs. Starting again with
the proof that bounded size important separators are bounded in number, we
wish to highlight two further main contributions in this vein. First,
Marx et al.~\cite{cuts-treewidth} showed that all inclusion-minimal
$st$-cuts of size at most $k$ in a graph $G$ (for fixed vertices $s, t$) 
are contained in (essentially) a subgraph of $G$ of treewidth at most $f(k)$. 
This statement is attractive in its simplicity, although algorithmic
applications of it tend to have a bad running time dependency on $k$
(since $f(k)$ is already exponential, and a treewidth dynamic programming
algorithm would normally have running time exponential in $f(k)$). 
Second, the FPT algorithm for \textsc{Minimum Bisection}~\cite{bisection-alg}
involves a technically intricate but very powerful decomposition of the input
graph in a tree-like manner into ``nearly unbreakable'' parts
(see the paper for details). However, these two results apply only to
undirected graphs, and for directed graphs we have a much weaker grasp of the structure.

A central milestone in the above developments was the discovery in 2010 of
FPT algorithms for \textsc{Multicut}, a robust generalization of
\textsc{Multiway Cut}, parameterized by the size of the cutset
only~\cite{bousquet,marx-razgon}. In this problem we are given a graph
$G$, a family $\terms$ of $k$ terminal pairs, and an integer $p$; the goal
is to delete at most $p$ nonterminal vertices from the graph such that all
terminal pairs are separated.\footnote{In this paper we use consistently
  the node-deletion variant of \textsc{Multicut}; note that it is
  equivalent to the edge-deletion setting in the directed case, and
  strictly more general in the undirected case.}
Marx and Razgon~\cite{marx-razgon} also proved that a similar result is unlikely
in directed graphs, showing that \textsc{Directed Multicut} is $W[1]$-hard parameterized
by the cutset only. In contrast, \textsc{Directed Multiway Cut} is FPT
by the same parameter~\cite{dmwc-alg}. 
Marx and Razgon~\cite{marx-razgon} asked about the complexity of \textsc{Directed Multicut}
when you restrict also the number of terminal pairs, either assuming it is constant
or by adding it to the parameter. 
Since then, this has become one of the most important open problems in the study
of graph separation problems in parameterized complexity.
It was positively resolved for DAGs (i.e., the problem is FPT on directed
acyclic graphs parameterized by both the cutset and the number of terminals)~\cite{dags-alg},
but otherwise little progress has been made.
In particular, it has been repeated in the survey of Marx~\cite{festschrift}
on future directions of parameterized complexity, where also the study of cut
problems in directed graphs has been identified as an important research direction.
Similarly, the question of the structure of bounded-size directed cuts (e.g., the existence 
of a directed version of~\cite{cuts-treewidth}) has been floating around the community.

\paragraph{Our results.}
We almost completely resolve the question of Marx and Razgon~\cite{marx-razgon}:

\begin{theorem}\label{thm:dirmc-lb}
\dirmc{} is $W[1]$-hard when parameterized by $p$, the size of the cutset,
even if restricted to instances with only four terminal pairs.
Furthermore, assuming the Exponential Time Hypothesis (ETH), there is no algorithm
solving $n$-vertex instances of \dirmc{} with four terminal pairs in time
$f(p) n^{o(p / \log p)}$ for any computable function $f$.
\end{theorem}

Since \dirmc{} is polynomial-time solvable for one terminal pair, and reduces
to \textsc{Directed Multiway Cut} for two terminal pairs~\cite{dmwc-alg},
only the case of three terminal pairs remains open.

The Exponential Time Hypothesis (ETH)~\cite{eth}, now a standard
lower bound assumption in parameterized complexity and moderately-exponential
algorithms (cf.~\cite{bulletin}), essentially asserts
that there exists no algorithm verifying satisfiability
of 3-CNF SAT formulae in time subexponential in the number of variables.
The ETH lower bound of Theorem~\ref{thm:dirmc-lb} implies that a brute-force
solution running in time $n^{\Oh(p)}$ is close to optimal.

Furthermore, we observe that our gadgets and reduction outline, after minor modifications,
gives also a lower bound for the \stor{} problem, answering an open question
of~\cite{stor}. In the \stor{} problem, given a mixed graph $G$ with $k$ terminal
pairs $\terms$, one asks for an orientation of all undirected edges of $G$, such 
that for every terminal pair $(s,t) \in \terms$ there is a path from $s$ to $t$
in the oriented graph.

\begin{theorem}\label{thm:stor-lb}
\stor{} is $W[1]$-hard when parameterized by $k$, the number of terminal pairs.
Furthermore, assuming ETH, there is no algorithm solving $n$-vertex
instances of \stor{} in time $f(k) n^{o(k / \log k)}$ for any computable function
$f$.
\end{theorem}

Cygan, Kortsarz, and Nutov~\cite{stor} showed an $n^{\Oh(k)}$-time algorithm;
Theorem~\ref{thm:stor-lb} implies that there is little hope for a fixed-parameter algorithm, and that the running time of the algorithm of Cygan, Kortsarz, and
Nutov is close to optimal.

We also show a structural result for directed $st$-cuts of bounded size. A
graph $G$ is \emph{$k$-cut-minimal} (for arc-cuts) if it contains a source
vertex $s$ and a sink vertex $t$, and every other vertex is incident to at
least one arc that participates in an inclusion-minimal $st$-cut of size
at most $k$. (For example, one can reduce an arbitrary graph to a $k$-cut-minimal
one by repeatedly bypassing every vertex that is not incident to any such arcs.) 
We have the following.

\begin{theorem}\label{thm:dtw-bound}
  Every $k$-cut-minimal graph has directed treewidth at most $f(k)$, where
  $f(k)=2^{\Oh(k^2)}$.
\end{theorem}

We note that this result is likely not as tight as it could be, both in
terms of the choice of width measure and the function $f(k)$. However, 
as discussed below, our results also show that \dirmc{} is essentially
hard already for four terminal pairs and directed pathwidth two, so for the
purpose of the \dirmc{} problem, a sharpening of Theorem~\ref{thm:dtw-bound} 
is not likely to help (even though it may be useful for other problems).

\paragraph{Discussion.}
In the presence of the lower bound of Theorem~\ref{thm:dirmc-lb}, whose
proof in this manuscript takes less than five pages including a figure and
is not very technically complex, the natural question is why the lower
bound was so elusive in the past five years, despite significant effort.

The idea to represent an $n$-wise choice in the reduction (e.g., the
choice of a vertex of a clique or an image of a vertex of the pattern graph
in the case of the \textsc{Subgraph Isomorphism} problem)
as a choice of a single-vertex cut on a
$n$-vertex bidirectional path connecting a terminal pair, seems quite
natural, especially in the light of the FPT algorithm for \dirmc{} in
DAGs~\cite{dags-alg} and a deeper study of where
this algorithm uses the assumption of acyclicity. However, a gadget to
check an arbitrary binary relation between two such choices was elusive;
we stumbled upon a ``correct'' construction by investigating the
structural question that led to Theorem~\ref{thm:dtw-bound} (the idea
being that a sufficiently powerful structural result could have been
useful in designing an efficient algorithm).

A few words about directed width measures may be in order (for more, see,
e.g.,~\cite{kreutzer-dimeas,kreutzer-width-chapter}). Two approaches
are possible. On the one hand, one can simply consider the underlying
undirected graph and ignore the arc orientations; e.g., if the underlying
undirected graph of a directed graph $G$ has bounded treewidth, then many
problems are solvable on $G$ via standard dynamic programming. On the
other hand, several width notions specific to directed graphs have been
proposed (to name a few: directed pathwidth; directed treewidth; DAG-width; Kelly-width).
Of these, directed treewidth is the most general and directed pathwidth
the most restrictive. It has been noted that these directed width notions
are rather more difficult to use in algorithms than the undirected
variants (e.g.,~\cite{GHKLOR-diwidth,GanianHKMORS10}), but on the other hand
they are more permissive (e.g., they all include DAGs as a constant-width base case).

\begin{figure*}[tb]
\begin{center}
\includegraphics{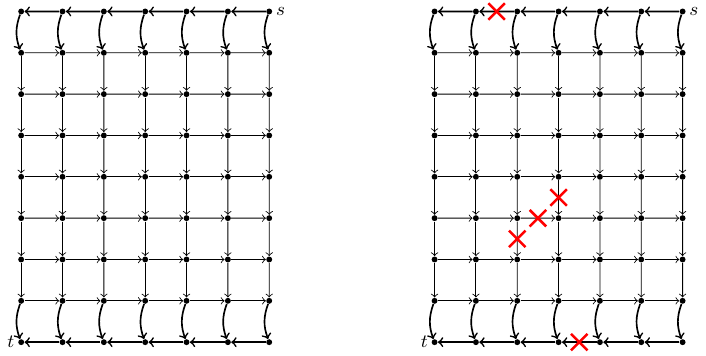}
\caption{An example of $5$-cut-minimal graph with unbounded treewidth of the underlying undirected graph.
The right panel shows a minimal cut of size five;
observe that every edge of the graph participates in a similar cut, except for a few edges in the four corners of the graph,
which can be easily covered by other small cuts.}
\label{fig:ex}
\end{center}
\end{figure*}

Our gadgets are based around a counterexample against the first of these
approaches. Figure~\ref{fig:ex} shows a graph which is 5-cut-minimal (in
fact, where \emph{every} arc participates in a minimal $st$-cut of bounded
size), but whose underlying undirected graph is a grid (and thus has
unbounded treewidth). Our hardness reduction uses this graph to route one
commodity horizontally and one vertically, allowing us to control which
pairs of flows can be killed at unit cost.

Also note that the graph in Figure~\ref{fig:ex} is acyclic; in fact, 
%it is not hard to show 
%%(e.g., using the cops and robber game characterization of directed pathwidth~\cite{dpw-game1}) 
%that
the reduction of Theorem~\ref{thm:dirmc-lb} 
outputs an instance of integer-weighted \dirmc{} of directed pathwidth 2.
Standard reductions (replacing a vertex of weight $w$ by $w$ unit-weight copies)
creates an instance of unweighted \dirmc{} with directed pathwidth bounded
polynomially in the parameter. Hence \dirmc{} remains $W[1]$-hard, even for four terminal pairs, 
if we parameterize by both the size of the cutset and the directed pathwidth of the input graph.
By the discussion above, this leaves little room for generalizing the FPT-algorithm for 
DAGs~\cite{dags-alg} to broader classes. 

%This discussion also leaves two open questions regarding
%Theorem~\ref{thm:dtw-bound}, namely, whether the results can be sharpened
%(e.g., to directed pathwidth) and what the best possible dependency $f(k)$
%is. We leave both questions for future work.

\paragraph{Organization and notation.}
We show the reduction for Theorem~\ref{thm:dirmc-lb} in Section~\ref{sec:hardness}
and for Theorem~\ref{thm:stor-lb} in Section~\ref{sec:stor}. 
We prove Theorem~\ref{thm:dtw-bound} in Section~\ref{sec:dtw-bound}.

We use standard graph notation, see e.g.~\cite{digraphs}.
Both our hardness reductions start from the $W[1]$-hard 
\subiso{} problem, parameterized by the number
of edges of the pattern graph.
An input to \subiso{} consists of two graphs $G$ and $H$ with $|E(H)| = k$,
where $V(G)$ is partitioned into $|V(H)|$ color classes, one for each
vertex of $H$: $V(G) = \biguplus_{i \in V(H)} V_i$;
the goal is to check if there exists a homomorphism
$\phi : V(H) \to V(G)$ such that $\phi(i) \in V_i$ for every $i \in V(H)$.
The $W[1]$- and ETH-hardness of \subiso{} has been shown by Marx:
\begin{theorem}[Cor.~6.3 of~\cite{subiso-lb}]\label{thm:subiso-lb}
\subiso{} is $W[1]$-hard when parameterized by $k = |E(H)|$.
Furthermore, assuming ETH, there is no algorithm solving $n$-vertex instances
of \subiso{} in time $f(k) n^{o(k / \log k)}$ for any computable
function $f$.
\end{theorem}
Both our reductions work in polynomial time and, given a \subiso{}
instance $(G,H)$,
output an instance with the corresponding parameter bounded linearly
in $k = |E(H)|$ and size bounded polynomially in the input size.

In both reductions, given an input instance $(G,H)$ of \subiso{}, we denote
$V(G) = \bigcup_{i \in V(H)} V_i$, $k = |E(H)|$, and $\ell = |V(H)|$.
Furthermore, without loss of generality we assume that $H$ does not have any isolated
vertices and thus $\ell \leq 2k$. Indeed,
if some color class $V_i$ of $V(G)$ is empty,
we can output a trivial no-instance,
and otherwise we can delete all isolated vertices from $H$ and their
corresponding color classes in $G$ without changing the answer to the problem.
We also fix some total order $<$ on the set $V(H)$ and, by some potential padding,
assume that $|V_i| = n$ and $V_i = \{v_a^i : 1 \leq a \leq n \}$
 for every $i \in V(H)$.

For discussion on the output of the reduction of Theorem~\ref{thm:dirmc-lb}, we need to recall the definition of directed pathwidth;
we follow here Kreutzer and Ordyniak~\cite{kreutzer-width-chapter}. 
\begin{definition}
Let $G$ be a directed graph. A \emph{directed path decomposition}
is a pair $(P,\beta)$ where $P = v_1,v_2,\ldots,v_s$ is a directed path
and $\beta: V(P) \to 2^{V(G)}$ is a mapping satisfying:
\begin{enumerate}
\item For every $v \in V(G)$ the set $\beta^{-1}(v)$ is connected in $P$
and nonempty.
\item For every arc $(u,v) \in E(G)$ there are indices $1 \leq i \leq j \leq s$
such that $u \in \beta(v_i)$ and $v \in \beta(v_j)$.
\end{enumerate}
The \emph{width} of $(P,\beta)$ is $\max_{1 \leq i \leq s} |\beta(v_i)|$,
and the \emph{directed pathwidth} of $G$ is the minimum width among all its
directed path decompositions.
For every $1 \leq i \leq s$, the set $\beta(v_i)$ is called a \emph{bag}
at node $v_i$.
\end{definition}
Note that directed pathwidth of a nonempty 
directed acyclic graph $G$ equals $1$:
we can take vertices of $G$ one-by-one in any topological ordering of $G$.

\section{Hardness for \dirmc}\label{sec:hardness}
\newcommand{\sincx}{s_{0 \to n}^{x}}
\newcommand{\tincx}{t_{0 \to n}^{x}}
\newcommand{\sincy}{s_{0 \to n}^{y}}
\newcommand{\tincy}{t_{0 \to n}^{y}}
\newcommand{\sdeclt}{s_{n \to 0}^{<}}
\newcommand{\tdeclt}{t_{n \to 0}^{<}}
\newcommand{\sdecgt}{s_{n \to 0}^{>}}
\newcommand{\tdecgt}{t_{n \to 0}^{>}}

We reduce an input \subiso{} instance $(G,H)$
to a node-weighted variant of \dirmc{}, where every non-terminal vertex has a weight being a positive integer, and the goal is to find a cutset of total weight
not exceeding the budget $p$.
We fix some integer constant $M$; in fact it suffices to set $M=2$,
but it helps to think of it as a sufficiently large integer constant.
 We use three levels of weight: there will be \emph{light vertices}, of weight $1$, \emph{heavy vertices}, of various integer weights being multiples of $M$,
usually depending on the degree of the corresponding vertex of $H$, and \emph{undeletable vertices}, of weight $p+1$. Observe that it is easy to reduce the weighted variant to the original one by replacing every vertex of weight $w$ with $w$ unit-weight vertices. Thus, it suffices to show the reduction to the weighted variant, with only four terminal pairs.

\paragraph{Construction.}
Let us now describe the construction of the (weighted)
\dirmc{} instance $(G',\terms,p)$.
We start by setting budget $p = (6M+1)k$.
We also introduce eight terminals, arranged in four terminal pairs:
\begin{equation*}
\begin{split}
(\sincx,\tincx),\quad (\sincy,\tincy), \\ \quad (\sdeclt,\tdeclt), \quad (\sdecgt,\tdecgt).
\end{split}
\end{equation*}
 
For every $i \in V(H)$, we introduce a bidirected path on $2n+1$ vertices
$$z^i_0, \hat{z}^i_1, z^i_1, \hat{z}^i_2, z^i_2, \ldots, \hat{z}^i_n, z^i_n,$$
called henceforth the \emph{$z$-path for vertex $i$}, and denoted by $Z^i$.
Similarly, for every ordered pair $(i,j)$ where $ij \in E(H)$, we introduce two bidirected paths on $2n+1$ vertices
\begin{equation*}
\begin{split}
x^{i,j}_0, \hat{x}^{i,j}_1, x^{i,j}_1, \hat{x}^{i,j}_2, x^{i,j}_2, \ldots, \hat{x}^{i,j}_n, x^{i,j}_n;\\
y^{i,j}_0, \hat{y}^{i,j}_1, y^{i,j}_1, \hat{y}^{i,j}_2, y^{i,j}_2, \ldots, \hat{y}^{i,j}_n, y^{i,j}_n.
\end{split}
\end{equation*}
We call these paths \emph{the $x$-path and the $y$-path for the pair $(i,j)$}, and denote them by $X^{i,j}$ and $Y^{i,j}$.
All vertices $z^i_a$, $x^{i,j}_a$, $y^{i,j}_a$ are undeletable, while all vertices $\hat{z}^i_a$, $\hat{x}^{i,j}_a$, $\hat{y}^{i,j}_a$ are heavy:
vertex $\hat{z}^i_a$ has weight $M \cdot \deg_H(i)$ and the vertices
$\hat{x}^{i,j}_a$ and $\hat{y}^{i,j}_a$ have weight $M$ each.

Note that so far we have created $4k + \ell$ paths, each having $2n+1$ vertices.
Furthermore, if we are to delete one heavy vertex from each of these paths,
the total cost would be
$$M \cdot \left(4k + \sum_{i \in V(H)} \deg_H(i)\right) = 6kM.$$

For every pair $(i,j)$ with $ij \in E(H)$,
and every $0 \leq a \leq n$, we add arcs $(x^{i,j}_a,z^i_a)$ and $(z^i_a,y^{i,j}_a)$. Furthermore, we attach terminals to the paths as follows.
\begin{itemize}
\item for every pair $(i,j)$ with $ij \in E(H)$,
we add arcs $(\sincx,x^{i,j}_0)$ and $(y^{i,j}_n,\tincy)$;
\item for every $i \in V(H)$ we add arcs $(\sincy,z^i_0)$ and $(z^i_n,\tincx)$; %, $(\sdeclt,z^i_n)$, $(\sdecgt,z^i_n)$, $(z^i_0,\tdeclt)$, and $(z^i_0,\tdecgt)$;
\item for every pair $(i,j)$ with $ij \in E(H)$ and $i<j$ we add arcs $(\sdeclt, x^{i,j}_n)$ and $(y^{i,j}_0,\tdeclt)$;
\item for every pair $(i,j)$ with $ij \in E(H)$ and $i > j$
we add arcs $(\sdecgt, x^{i,j}_n)$ and $(y^{i,j}_0,\tdecgt)$.
\end{itemize}

\begin{figure*}[tb]
\begin{center}
\includegraphics{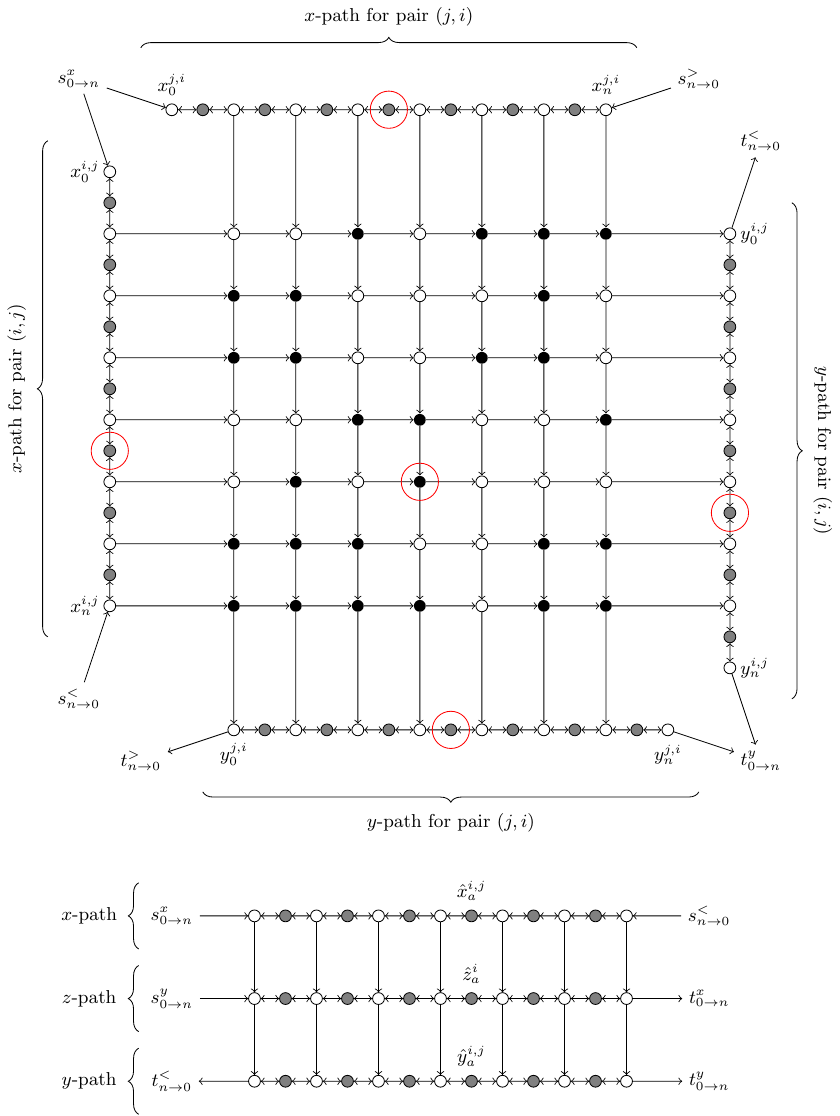}
\caption{Illustration of the reduction for \dirmc{}. Black vertices are light, gray are heavy, and white are undeletable.
The top figure illustates a $p$-grid, together with an intended solution marked by red circles.
Here, the vertex $p_{1,1}$ lies in the top-left corner of the grid, the first coordinate describes the row of the grid,
and the second one the column.
The bottom figure illustrates an $x$-, $z$-, and $y$-path for a pair $(i,j)$ with $ij \in E(H)$ and $i < j$.}
\label{fig:dirmc}
\end{center}
\end{figure*}

We refer to Figure~\ref{fig:dirmc} for an illustration.
Intuitively, the so-far constructed bidirectional paths and terminals require to delete at least one heavy vertex per 
bidirectional path; the connections between paths ensure that for every $i \in V(H)$, we need to chose
one index $1 \leq \phi(i) \leq n$ and delete vertices $\hat{z}^i_{\phi(i)}$, $\hat{x}^{i,j}_{\phi(i)}$, and $\hat{y}^{i,j}_{\phi(i)}$, that is, cut all paths corresponding to the vertex $i$ at the same place. The choice of the index $\phi(i)$ corresponds
to the choice of the image of $i$ in the sought homomorphism.

Let us now introduce gadgets that check the edge relations between the chosen vertices.
For every pair $(i,j)$ with $ij \in E(H)$, $i < j$ we introduce an acyclic $n \times n$ grid with vertices $p^{i,j}_{a,b}$ for $1 \leq a,b \leq n$
and arcs $(p^{i,j}_{a,b}, p^{i,j}_{a+1,b})$ for every $1 \leq a < n$ and $1 \leq b \leq n$, as well as 
$(p^{i,j}_{a,b}, p^{i,j}_{a,b+1})$ for every $1 \leq a \leq n$ and $1 \leq b < n$.
We call this grid \emph{$p$-grid for the pair $(i,j)$}, and denote it by $P^{i,j}$.
We set the vertex $p^{i,j}_{a,b}$ to be a light vertex if $v^i_a v^j_b \in E(G)$, and undeletable otherwise.
Finally, for every $1 \leq a \leq n$ we introduce arcs:
$$(x^{i,j}_a, p^{i,j}_{a,1}),\quad (p^{i,j}_{a,n}, y^{i,j}_{a-1}),\quad (x^{j,i}_a,p^{i,j}_{1,a}), \quad (p^{i,j}_{n,a}, y^{j,i}_{a-1}).$$
Intuitively, after deleting the aforementioned heavy vertices on the $x$-, $y$-, and $z$-paths, for fixed $ij \in E(H)$, $i < j$,
there is only one remaining path from $x^{i,j}_n$ (an out-neighbor of $\sdeclt$) to $y^{i,j}_0$ (an in-neighbor of $\tdeclt$),
passing through the $\phi(i)$-th row of the $p$-grid for the pair $(i,j)$,
and there is only one remaining path from $x^{j,i}_n$ (an out-neighbor of $\sdecgt$) to $y^{j,i}_0$ (an in-neighbor of $\tdecgt$),
passing through the $\phi(j)$-th column of the $p$-grid for the pair $(i,j)$.
We can kill both these paths with a single vertex $p^{i,j}_{\phi(i),\phi(j)}$, but only
the existence of the edge $v^i_{\phi(i)} v^j_{\phi(j)}$ ensures that this vertex is light, not undeletable.

This concludes the construction of the instance $(G',\terms,p)$.
We now formally show that the constructed instance is equivalent to the input
\subiso{} instance $(G,H)$.

\paragraph{From a homomorphism to a cutset.}
Let $\phi:V(H) \to [n]$ be such that $i \mapsto v^i_{\phi(i)} \in V_i$
is a homomorphism of $H$ into $G$.
Define
\begin{equation*}
\begin{split}
X &= \{\hat{x}^{i,j}_{\phi(i)}, \hat{y}^{i,j}_{\phi(i)} : (i,j) \in V(H) \times V(H) \textrm{\ s.t.\ }ij \in E(H)\}
\cup \\ &\cup \{\hat{z}^i_{\phi(i)} : i \in V(H) \}
\cup \{p^{i,j}_{\phi(i),\phi(j)} : (i,j) \in V(H) \times V(H) \textrm{\ s.t.\ }ij \in E(H), i < j\}.
\end{split}
\end{equation*}
The total weight of the vertices in $X$ equals:
$$2k \cdot 2 \cdot M + \sum_{i \in V(H)} M \cdot \deg_H(i) + k = (6M+1)k = p.$$
Note that the fact that 
$p^{i,j}_{\phi(i),\phi(j)}$ is light for every $ij \in E(H)$, $i < j$
follows from the assumption that the vertices $i \mapsto v^i_{\phi(i)}$ 
is a homomorphism.
Consequently, $X$ is of weight exactly $p$. We now show that it is a multicut in $(G',\terms)$.

We start with a simple observation about the structure of the graph $G'$. While
the $x$-, $y$-, and $z$-paths are bidirected, they --- together with the $p$-grids --- are arranged in a DAG-like fashion.
That is, there are directed arcs from $X^{i,j}$ to $Z^i$, from $Z^i$ to $Y^{i,j}$, from $X^{i,j}$ and $X^{j,i}$ to $P^{i,j}$, and from $P^{i,j}$ to $Y^{i,j}$ and $Y^{j,i}$,
but all cycles in $G'$ are contained in one $x$-, $y$-, or $z$-path. 
 
Consider first the terminal pair $(\sincx, \tincx)$. The out-neighbors of $\sincx$ are the endpoints $x^{i,j}_0$ for every pair $(i,j)$ with $ij \in E(H)$;
the only in-neighbors of $\tincx$ are the endpoints $z^i_n$ for every $i \in V(H)$.
Thus, by the previous observation, the only paths from $\sincx$ to $\tincx$ in the graph $G'$ start by going to some vertex $x^{i,j}_0$, traverse $X^{i,j}$ up to some vertex $x^{i,j}_a$, use the arc $(x^{i,j}_a,z^i_a)$
to fall to $Z^i$, and then traverse $Z^i$ to the endpoint $z^i_n$. 
(In particular, there are no paths from $\sincx$ to $\tincx$ that contain a vertex of some grid $P^{i,j}$.)
However, all such paths for $a \geq \phi(i)$ are cut
by the vertex $\hat{x}^{i,j}_{\phi(i)} \in X$, while all such paths for $a < \phi(i)$ are cut by the vertex $\hat{z}^i_{\phi(i)} \in X$. Consequently,
the terminal pair $(\sincx, \tincx)$ is separated in $H \setminus X$.

A similar argument holds for the pair $(\sincy,\tincy)$. By the same reasoning, the only paths between $\sincy$ and $\tincy$ in the graph $H$
are paths that start by going to some vertex $z^i_0$, traverse $Z^i$ up to some vertex $z^i_a$, use the arc
$(z^i_a,y^{i,j}_a)$ for some $ij \in E(H)$ to fall to $Y^{i,j}$, and then continue along this $y$-path to the vertex $y^{i,j}_n$.
However, all such paths for $a \geq \phi(i)$ are cut by the vertex $\hat{z}^{i}_{\phi(i)} \in X$, while all such paths for $a < \phi(i)$ are cut by the vertex $\hat{y}^{i,j}_{\phi(i)} \in X$.

Let us now focus on the terminal pair $(\sdeclt,\tdeclt)$. Observe that there are two types of paths from $\sdeclt$ to $\tdeclt$ in the graph $H$.
The first type consists of paths that starts by going to some vertex $x^{i,j}_n$ where $i < j$, traverse $X^{i,j}$ up to some vertex $x^{i,j}_a$,
use the arc $(x^{i,j}_a,z^i_a)$ to fall to $Z^i$, then traverse $Z^i$ up to some vertex $z^i_b$, use
the arc $(z^i_b,y^{i,j'}_b)$ for some $j' > i$ to fall to $Y^{i,j'}$, and finally traverse this $y$-path to the endpoint $y^{i,j'}_0$.
However, similarly as in the previous cases, the vertices $\hat{x}^{i,j}_{\phi(i)}, \hat{z}^i_{\phi(i)}, \hat{y}^{i,j'}_{\phi(i)} \in X$ cut all such paths.

The second type of paths use the $p$-grids in the following manner: the path starts by going to some vertex $x^{i,j}_n$ where $i < j$,
traverse $X^{i,j}$ up to some vertex $x^{i,j}_a$, use the arc $(x^{i,j}_a,p^{i,j}_{a,1})$ to fall to $P^{i,j}$,
traverse this $p$-grid up to a vertex $p^{i,j}_{b,n}$ where $b \geq a$, use the arc $(p^{i,j}_{b,n}, y^{i,j}_{b-1})$ to fall to $Y^{i,j}$,
and traverse this path to the endpoint $y^{i,j}_0$. These paths are cut by $X$ as follows: the paths where $a < \phi(i)$ are cut by $\hat{x}^{i,j}_{\phi(i)} \in X$,
the paths where $b > \phi(i)$ are cut by $\hat{y}^{i,j}_{\phi(i)}$, while the paths where $a=b=\phi(i)$ are cut by the vertex $p^{i,j}_{\phi(i),\phi(j)} \in X$;
note that the $\phi(i)$-th row of the grid is the only path from $p^{i,j}_{\phi(i),1}$ to $p^{i,j}_{\phi(i),n}$.
Please observe that the terminal $\tdeclt$ cannot be reached from $P^{i,j}$ by going to the other $y$-path reachable from this $p$-grid,
namely $Y^{j,i}$, as $Y^{j,i}$ has only outgoing arcs to the terminal $\tdecgt$ since $j > i$.

A similar argument holds for the pair $(\sdecgt,\tdecgt)$. The paths going through $X^{i,j}$, $i > j$,
$Z^i$, and $Y^{i,j'}$, $i > j'$, are cut by vertices
$\hat{x}^{i,j}_{\phi(i)}, \hat{z}^i_{\phi(i)}, \hat{y}^{i,j'}_{\phi(i)} \in X$.
The paths going through $X^{i,j}$, $i > j$, $P^{j,i}$, and $Y^{i,j}$,
are cut by the vertices $\hat{x}^{i,j}_{\phi(i)}, \hat{y}^{i,j}_{\phi(i)}, p^{j,i}_{\phi(j),\phi(i)}$. Again, it is essential that the other $y$-path
reachable from the $p$-grid for the pair $(j,i)$, namely $Y^{j,i}$, does not have outgoing arcs to the terminal $\tdecgt$, but only
to the terminal $\tdeclt$.

We infer that $X$ is a solution to the \dirmc{} instance $(G',\terms,p)$.

\paragraph{From a multicut to a homomorphism.}
Let $X$ be a solution to the \dirmc{} instance $(G',\terms,p)$. Our goal is to find a homomorphism of $H$ into $G$.

First, let us focus on heavy vertices in $X$. Observe that for every pair $(i,j)$, $ij \in E(H)$ the following three paths needs to be cut by $X$:
\begin{itemize}
\item
a path from $\sincx$ to $\tincx$ that traverses the entire path $X^{i,j}$ up to the vertex $x^{i,j}_n$, and uses the arc $(x^{i,j}_n,z^i_n)$ to reach
$\tincx$, 
\item a path from $\sincx$ to $\tincx$ that starts with using the arc $(x^{i,j}_0,z^i_0)$, and then traverses $Z^i$
up to the vertex $z^i_n$, and 
\item a path from $\sincy$ to $\tincy$ that starts with using the arc $(z^i_0,y^{i,j}_0)$, and then traverses $Y^{i,j}$
up to the vertex $y^{i,j}_n$.
\end{itemize}
We infer that $X$ needs to contain at least one heavy vertex on every $x$-, $y$-, and $z$-path in $H$. Recall that the total weight of exactly one heavy vertex from 
each of these paths is $6kM = p-k$. Thus, we have only $k$ slack in the budget constraint.

We say that a path $X^{i,j}$, $Y^{i,j}$, or $Z^i$ is \emph{normal}
if it contains exactly one vertex of $X$, and \emph{cheated}
otherwise.
We say that a pair $(i,j)$ for $ij \in E(H)$ is \emph{normal}
if each of the paths $X^{i,j}$, $Y^{i,j}$, $X^{j,i}$, $Y^{j,i}$, 
$Z^i$, and $Z^j$ is normal.
A pair $(i,j)$ is \emph{cheated} if it is not normal.
Note that $(i,j)$ is normal if and only if $(j,i)$ is normal.

For every $i \in V(H)$ we fix one $\phi(i) \in [n]$ such that
$\hat{z}^i_{\phi(i)} \in X$.

Fix now a normal pair $(i,j)$, $ij \in E(H)$.
Assume $i < j$; a symmetrical argument holds for $i > j$ but uses the terminal pair $(\sdecgt,\tdecgt)$
instead of $(\sdeclt,\tdeclt)$. Let $\hat{x}^{i,j}_a, \hat{z}^i_b, \hat{y}^{i,j}_c \in X$. 
Observe that $a \leq b$, as otherwise the path from $\sincx$ to $\tincx$ that traverses $X^{i,j}$ up to the vertex $x^{i,j}_{a-1}$, uses the arc $(x^{i,j}_{a-1},z^i_{a-1})$, and traverses $Z^i$ up to the endpoint $z^i_n$ is not cut by $X$, a contradiction.
A similar argument for the terminal pair $(\sincy,\tincy)$ implies that $b \leq c$.
However, if $a < b$, then the path from $\sdeclt$ to $\tdeclt$ that traverses $X^{i,j}$ up to the vertex $x^{i,j}_a$, uses the arc
$(x^{i,j}_a,z^i_a)$, traverses $Z^i$ up to the endpoint $z^i_0$, and finally uses the arc $(z^i_0,y^{i,j}_0)$, is not cut by $X$,
a contradiction. A similar argument implies gives a contradiction if $b < c$.

We infer that for every normal pair $(i,j)$ we have $\hat{z}^i_{\phi(i)} \in X$ and
$\hat{x}^{i,j}_{\phi(i)},\hat{y}^{i,j}_{\phi(i)} \in X$.

Fix now a normal pair $(i,j)$ with $i < j$. Observe that the following paths are not cut by the heavy vertices in $X$:
\begin{itemize}
\item a path from $\sdeclt$ to $\tdeclt$
that traverses $X^{i,j}$ up to the vertex $x^{i,j}_{\phi(i)}$, uses the arc $(x^{i,j}_{\phi(i)},p^{i,j}_{\phi(i),1})$,
traverses the $\phi(i)$-th row of $P^{i,j}$ up to the vertex $p^{i,j}_{\phi(i),n}$, uses the arc $(p^{i,j}_{\phi(i),n}, y^{i,j}_{\phi(i)-1})$, and traverses $Y^{i,j}$ up to the endpoint $y^{i,j}_0$;
\item a path from $\sdecgt$ to $\tdecgt$
that traverses $X^{j,i}$ up to the vertex $x^{j,i}_{\phi(j)}$, uses the arc $(x^{j,i}_{\phi(j)},p^{i,j}_{1,\phi(j)})$,
traverses the $\phi(j)$-th column of $P^{i,j}$ up to the vertex $p^{i,j}_{n,\phi(j)}$, uses the arc $(p^{i,j}_{n,\phi(j)}, y^{j,i}_{\phi(j)-1})$, and traverses $Y^{j,i}$ up to the endpoint $y^{j,i}_0$.
\end{itemize}
Consequently, $X$ needs to contain at least one light vertex in the $p$-grid
$P^{i,j}$.
Furthermore, if $X$ contains exactly one light vertex in $P^{i,j}$,
then, as the only vertex in common of the two aforementioned paths for a fixed choice of normal $(i,j)$, $i < j$, is the vertex $p^{i,j}_{\phi(i),\phi(j)}$,
we have that $p^{i,j}_{\phi(i),\phi(j)} \in X$ is a light vertex, and, by construction, $v^i_{\phi(i)} v^j_{\phi(j)} \in E(G)$.

It remains to show that every pair $(i,j)$, $ij \in E(H)$, is normal.
Indeed, if this is the case, then, as $p = 6kM + k$,
the total weight of exactly one heavy vertex on each path $X^{i,j}$, $Y^{i,j}$, and $Z^i$ is $6kM$,
and there are exactly $k$ grids $P^{i,j}$, every grid $P^{i,j}$ contains
exactly one vertex of $X$, and the argumentation from the previous section
shows that $i \mapsto v^i_{\phi(i)}$ is a homomorphism, concluding
the proof of Theorem~\ref{thm:dirmc-lb}.

Recall that $(i,j)$ is normal if and only if $(j,i)$ is normal.
Let $c$ be the number of cheated pairs $(i,j)$, $i < j$.
If $(i,j)$, $i < j$, is cheated, then there is a witness for it:
one of the paths $X^{i,j}$, $Y^{i,j}$, $X^{j,i}$, $Y^{j,i}$,
$Z^i$, or $Z^j$ is cheated, that is, contains more than one vertex of $X$.
However, note that a cheated path $X^{i,j}$, $Y^{i,j}$, $X^{j,i}$, or $Y^{j,i}$ is a witness only that $(i,j)$ and $(j,i)$ is cheated.
Let $c_{XY}$ be the number of cheated $x$- and $y$-paths.
Furthermore, a cheated path $Z^i$ is a witness that $(i,j)$ is cheated
for every $ij \in E(H)$: there are only $\deg_H(i)$ such pairs $(i,j)$.
We infer that
$$c \leq c_{XY} + \sum_{i \in V(H): Z^i\ \textrm{cheated}} \deg_H(i).$$
On the other hand, 
a cost of a second heavy vertex in $X$ on an $x$- or $y$-path is $M$,
while the cost of a second heavy vertex on $Z^i$ is $M \cdot \deg_H(i)$.
Furthermore, recall that if $(i,j)$, $i < j$, is normal, then $P^{i,j}$ contains at least one vertex of $X$.
Thus, the total weight of $X$ is at least
\begin{align*}
& 6kM + M\cdot c_{XY} + \sum_{i \in V(H): z^i\ \textrm{cheated}} M \cdot deg_H(i) + k - c \\
&\quad \geq 6kM + M\cdot c + k - c = p + (M-1)c.
\end{align*}
Consequently, if $M > 1$ and the weight of $X$ is at most $p$, we have $c=0$.
This finishes the proof of Theorem~\ref{thm:dirmc-lb}.

The resulting digraph can easily be shown to have directed pathwidth 2.
Clearly, it cannot have smaller directed pathwidth, as it contains two-vertex
cycles.
In the other direction, order the vertices of the resulting graphs as follows.
\begin{enumerate}
\item All source terminals.
\item The vertices $x^{i,j}_a$, sorted first by the pair $(i,j)$ lexicographically, and then by the subscript $a$.
\item The vertices $z^i_a$, sorted first by $i$, and then by the subscript $a$.
\item The vertices $p^{i,j}_{a,b}$, sorted first by the pair $(i,j)$ lexicographically, and then by the pair $(a,b)$ lexicographically.
\item The vertices $y^{i,j}_a$, sorted first by the pair $(i,j)$ lexicographically, and then by the subscript $a$.
\item All sink terminals.
\end{enumerate}
Observe now that we can 
construct a directed path decomposition of the resulting graph
by taking bags consisting of every two consecutive vertices in this order.

We infer that \dirmc{} is $W[1]$-hard already for integer-weighted instances,
parameterized by total solution weight, for instances with 4 terminals 
and directed pathwidth 2. This is in sharp contrast to the result
that it is FPT for DAGs~\cite{dags-alg}.

\section{Hardness for \stor}\label{sec:stor}
\begin{figure*}[tbh]
\begin{center}
\includegraphics{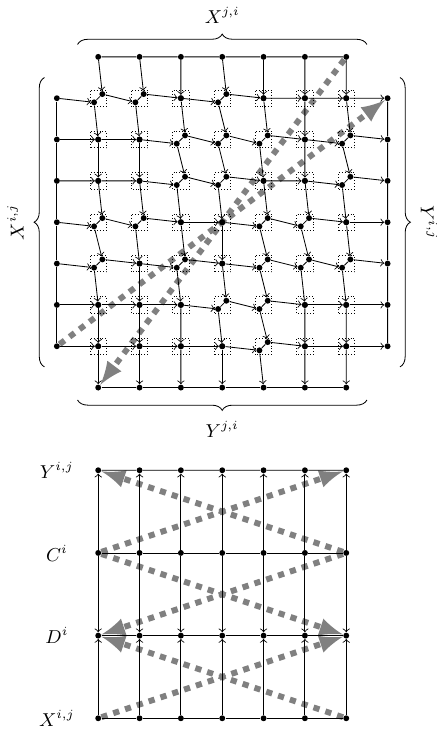}
\caption{Illustration of the reduction for \stor{}. Thick gray arrows represent terminal pairs.
The top figure illustates a $p$-grid; dashed rectangles represent single entities $p^{i,j}_{a,b}$.
The bottom figure illustrates synchronization between paths $X^{i,j}$, $Y^{i,j}$, $C^i$, and $D^i$.}
\label{fig:stor}
\end{center}
\end{figure*}

\paragraph{Construction.} 
Given an input \subiso{} instance $(G,H)$,
we construct an equivalent \stor{} instance $(G',\terms)$ as follows.

First, we introduce a number of undirected paths on $n$ vertices:
for every $i \in V(H)$ we introduce a path $C^i$ on vertices $c^i_1,c^i_2,\ldots,c^i_n$
and a path $D^i$ on vertices $d^i_1,d^i_2,\ldots,d^i_n$,
while for every ordered pair $(i,j)$ where $ij \in E(H)$, we introduce
a path $X^{i,j}$ on vertices $x^{i,j}_1,x^{i,j}_2,\ldots,x^{i,j}_n$
and a path $Y^{i,j}$ on vertices $y^{i,j}_1,y^{i,j}_2,\ldots,y^{i,j}_n$.
We connect these paths as follows:
for every $i \in V(H)$ and $1 \leq a \leq n$ we add an arc $(c^i_a,d^i_a)$,
while for every pair $(i,j)$ with $ij \in E(H)$ and for every $1 \leq a \leq n$ we add arcs
$(x^{i,j}_a,d^i_a)$ and $(c^i_a,y^{i,j}_a)$.
Furthermore, for every $i \in V(H)$ we add terminal pairs
$(c^i_1,d^i_n)$ and $(c^i_n,d^i_1)$
and for every pair $(i,j)$ where $ij \in E(H)$, we add terminal pairs
$(x^{i,j}_1,d^i_n)$, $(x^{i,j}_n,d^i_1)$, $(c^i_1,y^{i,j}_n)$, and $(c^i_n,y^{i,j}_1)$.

For an index $1 \leq a \leq n$, we say that the path $C^i$ is
\emph{oriented towards $c^i_a$} if for every $1 \leq b < n$, the edge $c^i_b c^i_{b+1}$ is oriented from $c^i_b$ to $c^i_{b+1}$ if $b < a$
and from $c^i_{b+1}$ to $c^i_b$ if $b \geq a$.
Similarly, we say that the path $D^i$ is \emph{oriented away from $d^i_a$} if for every $1 \leq b < n$, the edge $d^i_b d^i_{b+1}$ is oriented from $d^i_{b+1}$ to $d^i_{b}$ if $b < a$
and from $d^i_{b}$ to $d^i_{b+1}$ if $b \geq a$, and similar notions are defined for the paths $X^{i,j}$ and $Y^{i,j}$.

Informally speaking, the introduced paths play the role of $x-$, $y-$, and $z-$ paths
from the reduction for \dirmc{}.
In the currently constructed graph, the only way to satisfy both the terminal pair $(c^i_1,d^i_n)$ and the terminal pair $(c^i_n,d^i_1)$
is to chose a number $1 \leq \phi(i) \leq n$ and direct the path $C^i$ towards $c^i_{\phi(i)}$
and the path $D^i$ away from $d^i_{\phi(i)}$. 
The choice of $\phi(i)$ corresponds to the image of $i \in V(H)$
in the sought homomorphism.
The additional terminal pairs ensure that the choice is copied to other paths:
$X^{i,j}$ needs to be directed towards $x^{i,j}_{\phi(i)}$ and $Y^{i,j}$ needs to be directed away from $y^{i,j}_{\phi(i)}$.
In the remainder of the construction, we will not introduce any more edge nor arc incident to a vertex of any path $C^i$ or $D^i$,
and only arcs going out of the paths $X^{i,j}$ and arcs going towards the paths $Y^{i,j}$. In this way we will not introduce any unwanted
way of satisfying the terminal pairs introduced so far, guaranteeing the desired behavior.

Let us now proceed to the construction of the $p$-grids.
Similarly as in the case of \dirmc{},
for every pair $(i,j)$ with $ij \in E(H)$ and $i < j$ we introduce an acyclic $n \times n$ grid (called again a \emph{$p$-grid}) with vertices $p^{i,j}_{a,b}$ for $1 \leq a,b \leq n$
and arcs $(p^{i,j}_{a,b}, p^{i,j}_{a+1,b})$ for every $1 \leq a < n$ and $1 \leq b \leq n$, as well as 
$(p^{i,j}_{a,b}, p^{i,j}_{a,b+1})$ for every $1 \leq a \leq n$ and $1 \leq b < n$.
We connect the grid to the previously constructed
vertices as follows: for every $1 \leq a \leq n$ we introduce arcs
$$(x^{i,j}_a, p^{i,j}_{a,1}),\quad (p^{i,j}_{a,n}, y^{i,j}_{a}),\quad (x^{j,i}_a,p^{i,j}_{1,a}), \quad (p^{i,j}_{n,a}, y^{j,i}_{a});$$
please note the lack of shift by one as compared to the \dirmc{} construction.
Furthermore, we add terminal pairs $(x^{i,j}_n,y^{i,j}_1)$
and $(x^{j,i}_n,y^{j,i}_1)$.

Observe that, if the paths $X^{i,j}$ and $X^{j,i}$ were directed
towards $x^{i,j}_{\phi(i)}$ and $x^{j,i}_{\phi(j)}$ respectively,
while the paths $Y^{i,j}$ and $Y^{j,i}$ were directed away from
$y^{i,j}_{\phi(i)}$ and $y^{j,i}_{\phi(j)}$ respectively,
there are unique paths in the graph satisfying the newly introduced terminal pairs:
the one from $x^{i,j}_n$ to $y^{i,j}_1$ needs to traverse the grid
for the pair $(i,j)$ along the $\phi(i)$-th row, while
the one from $x^{j,i}_n$ to $x^{j,i}_1$ needs to traverse it
along the $\phi(j)$-th column. These two paths intersect
at the vertex $p^{i,j}_{\phi(i),\phi(j)}$;
we finish the construction by encoding
the edges of $G$ by the following modification of the vertices $p^{i,j}_{a,b}$.

For every vertex $p^{i,j}_{a,b}$, we call its incident four arcs as follows:
the \emph{north} arc goes from $p^{i,j}_{a-1,b}$ (or $x^{j,i}_b$ if $a=1$) to
$p^{i,j}_{a,b}$, the \emph{south} arc goes
from $p^{i,j}_{a,b}$ to $p^{i,j}_{a+1,b}$ (or $y^{j,i}_b$ if $a=n$),
the \emph{west} arc goes from $p^{i,j}_{a,b-1}$ (or $x^{i,j}_a$ if $b=1$)
to $p^{i,j}_{a,b}$, while the \emph{east} arc goes
from $p^{i,j}_{a,b}$ to $p^{i,j}_{a,b+1}$ (or $y^{i,j}_a$ if $b=n$).
If $v^i_a v^j_b \in E(G)$, we keep the vertex $p^{i,j}_{a,b}$
intact; otherwise, we split the vertex $p^{i,j}_{a,b}$
into two vertices $p^{i,j}_{a,b,SW}$, $p^{i,j}_{a,b,NE}$, connected by
an undirected edge, with $p^{i,j}_{a,b,SW}$
incident to the south and west arcs and $p^{i,j}_{a,b,NE}$ incident
to the north and east arcs. With this construction, we cannot traverse
the split vertex $p^{i,j}_{a,b}$
from north to south and from west to east at the same time, 
forbidding us from choosing $\phi(i) = a$ and $\phi(j) = b$
simultaneously.
This implies that $(G',\terms)$ is a positive instance if and only if 
$(G,H)$ is a positive instance.

This finishes the description of the constructed \stor{} instance
$(G',\terms)$. Note that the number of terminal pairs
is bounded by $\Oh(k + \ell) = \Oh(k)$. 
We now prove formally that the reduction
is a correct hardness reduction for \stor{}.

\paragraph{From a homomorphism to an orientation.}
Let $\phi:[k] \to [n]$ be such that $i \mapsto v^i_{\phi(i)}$ is a homomorphism
of $H$ into $G$.
We start by orienting some of the edges of $H$ as follows.
For every $i \in V(H)$, orient the path $C^i$ towards $c^i_{\phi(i)}$
and the path $D^i$ away from $d^i_{\phi(i)}$.
Similarly, for every pair $(i,j)$ with $ij \in E(H)$, orient
the path $X^{i,j}$ towards $x^{i,j}_{\phi(i)}$ and
orient the path $Y^{i,j}$ away from $y^{i,j}_{\phi(i)}$.

Let $G''$ be the mixed graph obtained so far.
For every terminal pair $(s,t) \in \terms$, we exhibit a path in $G''$
from $s$ to $t$ in such a way that no two such paths share an undirected
edge. This proves that the remaining undirected edges can be oriented
in the desired way.

It is straightforward to observe that for every terminal pair $(s,t)$
involving a $c$-vertex
or a $d$-vertex (i.e., one introduced in the first part of the construction)
there exists a directed path in $G''$ from $s$ to $t$ consisting of directed
arcs only: we may traverse from any endpoint of $X^{i,j}$ or $C^i$
up to $x^{i,j}_{\phi(i)}$
or $c^i_{\phi(i)}$, fall to $d^i_{\phi(i)}$ or $y^{i,j}_{\phi(i)}$,
and go along the path $D^i$ or $Y^{i,j}$ to any of its endpoints.
Thus, it remains to focus on the pairs $(x^{i,j}_n,y^{i,j}_1)$
and $(x^{j,i}_n,y^{j,i}_1)$ for $ij \in E(H)$, $i < j$.

For the first pair, we traverse $X^{i,j}$ up to $x^{i,j}_{\phi(i)}$,
then traverse along the $\phi(i)$-row of the grid
(i.e., using vertices $p^{i,j}_{\phi(i),b}$ for $1 \leq b \leq n$;
if any of these vertices is split, we traverse it from west to east
using the intermediate undirected edge), and traverse $Y^{i,j}$
from $y^{i,j}_{\phi(i)}$ up to $y^{i,j}_1$.
The path for the second pair is similar, but uses the $\phi(j)$-th
column of the grid, and traverses every vertex $p^{i,j}_{a,\phi(j)}$
for $1 \leq a \leq n$ from north to south.

These two paths intersect at $p^{i,j}_{\phi(i),\phi(j)}$.
Since $v^i_{\phi(i)} v^j_{\phi(j)} \in E(G)$, this vertex
is not split and the two aforementioned paths do not share an undirected edge.
As only these two paths are present in the grid for the pair $(i,j)$,
and all undirected edges in $G''$ are contained in such grids,
we conclude that $(G',\terms)$ is a yes-instance.

\paragraph{From an orientation to a homomorphism.}
Assume that $(G',\terms)$ is a yes-instance, and let $G''$ be an oriented graph $G'$ such that 
for every $(s,t) \in \terms$, there is a path from $s$ to $t$ in $G''$.

Fix a vertex $i \in V(H)$.
Recall that no arc leads toward any path $C^i$ or $X^{i,j}$, and no arc leads from any path $D^i$
or $Y^{i,j}$. Furthermore, all arcs leaving $C^i$ lead to $D^i$ or one of the paths $Y^{i,j}$,
and all arcs going into $D^i$ start in $C^i$ or in one of the paths $X^{i,j}$. We infer
that the path from $c^i_1$ to $d^i_n$ and the path from $c^i_n$ to $d^i_1$ in $G''$ both need to be completely
contained in $G''[C^i \cup D^i]$. Furthermore, such a path for $(c^i_1,d^i_n)$ traverses $C^i$ up to some vertex $c^i_a$,
falls to $d^i_a$, and continues along $D^i$, and similarly for the pair $(c^i_n,d^i_1)$. We infer that there exists
an index $1 \leq \phi(i) \leq n$ such that $C^i$ is oriented towards $c^i_{\phi(i)}$, while $D^i$
is oriented away from $d^i_{\phi(i)}$. 

We claim that $i \mapsto v^i_{\phi(i)}$ is a homomorphism of $H$ into $G$.
Fix an edge $ij \in E(H)$, $i < j$;
we aim to show that the vertex $p^{i,j}_{\phi(i),\phi(j)}$ is not split, which is equivalent to $v^i_{\phi(i)} v^j_{\phi(j)} \in E(G)$.

A similar reasoning as earlier for the path $X^{i,j}$ and terminal pairs 
$(x^{i,j}_1,d^i_n)$ and $(x^{i,j}_n,d^i_1)$ implies that $X^{i,j}$ is oriented towards $x^{i,j}_{\phi(i)}$.
Analogously, we obtain that $X^{j,i}$ is oriented towards $x^{j,i}_{\phi(j)}$,
$Y^{i,j}$ is oriented away from $y^{i,j}_{\phi(i)}$, and $Y^{j,i}$ is oriented away from $y^{j,i}_{\phi(j)}$.

The only arcs leaving the path $X^{i,j}$ are arcs going towards $D^i$ and the $p$-grid for the pair $(i,j)$.
Similarly, the only arcs ending in $Y^{i,j}$ start in $C^i$ and in the aforementioned $p$-grid, and analogous claims hold for the paths $X^{j,i}$ and $Y^{j,i}$.
We infer that the paths in $H'$ for terminal pairs $(x^{i,j}_n,y^{i,j}_1)$ and $(x^{j,i}_n,y^{j,i}_1)$ need to traverse through the $p$-grid for the pair $(i,j)$.

With the paths $X^{i,j}$, $X^{j,i}$, $Y^{i,j}$, and $Y^{j,i}$ oriented as described, even with keeping the undirected edges in the grid not oriented,
the only path for the first terminal pair traverses the $\phi(i)$-th row of the grid from $x^{i,j}_{\phi(i)}$, through
$p^{i,j}_{\phi(i),b}$ for $1 \leq b \leq n$ (split or not), towards $y^{i,j}_{\phi(i)}$, and similarly
the only path for the second terminal pair traverses the $\phi(j)$-th column of the grid. If
the vertex $p^{i,j}_{\phi(i),\phi(j)}$ is split, these paths traverse the undirected edge $p^{i,j}_{\phi(i),\phi(j),SW} p^{i,j}_{\phi(i),\phi(j),NE}$
in opposite directions, a contradiction.

This concludes the proof of Theorem~\ref{thm:stor-lb}.

\section{Directed treewidth bound of $k$-cut-minimal graphs}\label{sec:dtw-bound}
We now prove the directed treewidth upper bound
(Theorem~\ref{thm:dtw-bound}).
For the duration of this section, we will consider arc cuts for convenience,
but the result also implies a very similar statement for vertex cuts.

We will need some further
preliminaries.
For a vertex set $U$, we let $\delta(U)$ denote the set of arcs leaving $U$.
Let $G=(V,A)$ be a digraph. A set $T \subseteq V$ is \emph{well-linked} if
for every pair of equal sized subsets $X$ and $Y$ of $T$, there are $|X|$
vertex disjoint paths from $X$ to $Y$ in $G-(T \setminus (X \cup Y))$. 
Well-linked sets are connected to directed treewidth in the following sense.

\begin{theorem}[Cor.~6.4.24 of~\cite{kreutzer-width-chapter}] \label{thm:well-linked}
  Let $G$ be a digraph with no well-linked set of cardinality more than $k$.
  Then $G$ has directed treewidth $\Oh(k)$.
\end{theorem}
In what follows we will work only with well-linked sets, and thus we omit the definition of directed treewidth.
For the definition of directed treewidth, and 
more on directed width notions, see Kreutzer and Ordyniak~\cite{kreutzer-width-chapter}.

We will show that a $k$-cut-minimal graph cannot contain an arbitrarily large 
well-linked set. To illustrate the approach, assume that $T \subseteq V$
is a well-linked set of sufficiently large size 
($|T| > f(k)$ for some $f(k)$ yet to be specified). 
We will identify a set $U \subseteq V$ with $|\delta(U)|\leq k$, 
such that $|T \cap U|, |T \setminus U| > k$. This contradicts
the assumption that $T$ is well-linked, and implies that $G$
has bounded directed treewidth.

The basic engine of the proof is the \emph{anti-isolation lemma}, 
due to Marx. 
For this, we need to
recall the notion of important separators. (We use the version for directed 
graphs~\cite{dmwc-alg}, adapted to our purpose; in particular, we use arc 
separators instead of vertex separators. The original definition 
for undirected graphs was~\cite{marx-impsep}.)

Let $G=(V,A)$ be a digraph with $s, t \in V$. An \emph{important st-separator}
is a minimal $s$-$t$-cut $C \subseteq A$ such that there is no $s$-$t$-cut 
$C' \subseteq A$, $C' \neq C$, with $|C'| \leq |C|$ such that 
every vertex reachable from $s$ in $G-C$ is also reachable from $s$ in $G-C'$. 

\begin{lemma}[\cite{dmwc-alg}]
  There are at most $4^k$ important $s$-$t$-separators of size at most $k$.
\end{lemma}

It is important to note that every $s$-$t$-cut $C$ can be ``pushed'' to an
important $s$-$t$-separator $C'$ with $|C'| \leq |C|$, such that every vertex
reachable from $s$ in $G-C$ is also reachable from $s$ in $G-C'$.

We state and prove the anti-isolation lemma. (This is taken from a
set of lecture slides of Marx~\cite{anti-isolation}; to the best of our
knowledge, no proof has appeared in a formally reviewed publication.)
Because the expression will be used several times, define 
$g(k)=(k+1)4^{k+1}$. 

\begin{lemma}[Anti-isolation~\cite{anti-isolation}] \label{lm:anti-isolation}
  Let $s, v_1, \ldots, v_r$ be vertices in a digraph $G=(V,A)$, and let
  $C_1, \ldots, C_r \subseteq A$ be arc sets of size at most $k$
  such that for all $i, j \in [r]$ there is a path from $s$ to $v_j$
  in $G-C_i$ if and only if $i=j$. Then $r \leq g(k)$.
\end{lemma}
\begin{proof}
  Add a vertex $t$, and an arc $(v_i,t)$ for every $i \in [r]$.
  Then for every $i \in [r]$, the cut $C_i \cup \{v_i,t\}$ is an $s$-$t$-cut
  of size at most $k+1$, and can be pushed to an important separator $C_i'$
  of size at most $k+1$. Note that necessarily $(v_i,t) \in C_i'$. 
  Since there are only at most $4^{k+1}$ important separators of size
  at most $k+1$, and each of them contains at most $k+1$ arcs $(v_j,t)$,
  $j \in [r]$, and since every such arc is contained in an important separator,
  we conclude that $r \leq (k+1)4^{k+1}=g(k)$. 
\end{proof}

We will also use the following dual form.

\begin{lemma} \label{lm:reverse-anti}
  Let $t, v_1, \ldots, v_r$ be vertices in a digraph $G=(V,A)$, and let
  $C_1, \ldots, C_r \subseteq A$ be arc sets of size at most $k$
  such that for all $i, j \in [r]$ there is a path from $v_j$ to $t$
  in $G-C_i$ if and only if $i=j$. Then $r \leq g(k)$.
\end{lemma}
\begin{proof}
  Simply apply the anti-isolation lemma to the reversed graph $G'=(V, \{(v,u): (u,v) \in A\})$.
\end{proof}

Let us now recall the sunflower lemma. A \emph{sunflower} 
is a collection $\{X_1,\ldots, X_r\}$ of subsets of a ground set $V$
such that all pairwise intersections are identical, i.e.,
there is a set $K \subseteq V$ (the \emph{core} of the sunflower)
such that for all $i, j \in [r]$, $i \neq j$ we have
$X_i \cap X_j = K$. 
The (pairwise disjoint) sets $X_i \setminus K$ for $i \in [r]$ are called the \emph{petals} of the sunflower.
The famous Sunflower Lemma says the following
(original lemma is due to Erd\H{o}s and Rado~\cite{sunflower-original};
the following is from Flum and Grohe~\cite{flum-grohe}).

\begin{lemma} \label{lm:sunflower} 
  Let $\mathcal{H} \subseteq 2^V$ be a collection of subsets of size $d$
  of a ground set $V$. If $|\mathcal{H}|>d!k^d$, then $\mathcal{H}$
  contains a sunflower of cardinality more than $k$.
\end{lemma}

Observe that for any minimal $st$-cut $C$
containing an arc $(u,v)$, the graph $G-C$ contains both an $su$-path and a $vt$-path.
Using the Erd\H{o}s-Rado sunflower lemma, 
we can then sharpen these statements into the following.

\begin{lemma} \label{lemma:balancedcuts}
  Let $G=(V,A)$ be a digraph, with two distinguished vertices $s, t \in V$. 
  Let $F \subseteq A$ be a collection of arcs such that 
  every arc $(u,v) \in F$ participates in some minimal $st$-cut
  of size at most $k$. 
  %If there is no minimal $st$-cut $(S, V \setminus S)$ 
  %of size at most $k$ such that both $S$ and $V \setminus S$ contains
  %both endpoints of more than $k$ arcs from $F$, 
  %then $|F|$ is at most $2^{\Oh(k^2)}$. 
  If $|F|>h(k)$ where $h(k)=2^{\Oh(k^2)}$, then there is a minimal
  $st$-cut $C \subseteq A$ in $G$ of size at most $k$ that splits $F$, i.e., 
  there are more than $k$ arcs $(u,v) \in F \setminus C$ such that $G-C$ contains an $su$-path,
  and more than $k$ arcs $(u,v) \in F \setminus C$ such that $G-C$ contains a $vt$-path.
\end{lemma}
\begin{proof}
  For each $(u,v) \in F$, let $C_{uv} \subseteq A$ be a minimal
  $s$-$t$-cut with $(u,v) \in C_{uv}$, with $|C_{uv}|\leq k$.
  Let $S_{uv}$ denote the set of vertices reachable from $s$
  in $G-C_{uv}$, and let $T_{uv}$ denote the set of vertices
  that reach $t$ in $G-C_{uv}$. From the observation
  stated before the lemma statement,
  by the minimality of $C_{uv}$ we have $u \in S_{uv}$
  and $v \in T_{uv}$. 
  Also let $U=\{u: (u,v) \in F\}$ and $W=\{v: (u,v) \in F\}$. 
  We first make a claim to lower-bound the size of $U$ and $W$.
  (This is standard, but we include a proof for completeness.)

  \begin{myclaim}
    For every vertex $v$, at most $2g(k)$ arcs incident to $v$
    participate in minimal $s$-$t$-cuts of size at most $k$.
  \end{myclaim}
  \begin{proof}
    We bound the number of in-arcs and out-arcs separately.
    For the in-arcs, consider the graph resulting by adding
    $k+1$ arcs $(v,t)$. Note that every minimal $s$-$t$-cut
    that contains an in-arc of $v$ remains a minimal $s$-$t$-cut.
    As in the proof of the anti-isolation lemma, we may push
    each such cut to an important separator, and this operation
    will not decrease the set of in-arcs of $v$ contained in the cut.
    Hence at most $g(k)$ in-arcs of $v$ participate
    in minimal $s$-$t$-cuts of cardinality at most $k$. 
    The proof bounding out-arcs is entirely analogous. 
  \cqed\end{proof}

  We make a further subsidiary claim.

  \begin{myclaim}
    There are at most $2^{\Oh(k^2)}$ arcs $(u,v) \in F$
    such that $|S_{uv}\cap U| \leq 2k$.
  \end{myclaim}
  \begin{proof}
    Consider the collection 
    $\mathcal{S}=\{S_{uv} \cap U: (u,v) \in F\}$. 
    Let $X_1,X_2,\ldots,X_r$ be a sunflower in $\mathcal{S}$ with core $K$ and let $C_i$ be the cut $C_{uv}$ such that $S_{uv} \cap U = X_i$ for $i \in [r]$.
    Pick a vertex $x_i \in X_i \setminus K$ for every $i \in [r]$. Then, 
    the vertices $s$ and $x_i$ together with cuts $C_i$ satisfy the prerequisities for the anti-isolation lemma.
    Consequently, $\mathcal{S}$ does not contain a sunflower of cardinality more than $g(k)$. 

    Hence by the sunflower lemma, it contains at most 
    $d!(g(k))^d$ distinct sets of cardinality $d$, 
    for each $1 \leq d \leq 2k$. Summing, and ignoring the 
    constant factor, we find that $\mathcal{S}$ contains
    only $2^{\Oh(k^2)}$ sets of size at most $2k$. By the 
    degree observation in the previous claim, every $u \in U$ is an 
    endpoint for at most $g(k)$ arcs $(u,v) \in F$;
    hence there are at most $g(k)2^{\Oh(k^2)}=2^{\Oh(k^2)}$ 
    arcs $uv \in F$ such that $|S_{uv} \cap U|\leq 2k$,
    as promised.
  \cqed\end{proof}
  
  The dual bound also holds.

  \begin{myclaim}
    There are at most $2^{\Oh(k^2)}$ arcs $(u,v) \in F$
    such that $|T_{uv}\cap W| \leq 2k$.
  \end{myclaim}
  \begin{proof}
    This proof is identical to the previous one (except 
    using the form Lemma~\ref{lm:reverse-anti} instead
    of the usual anti-isolation lemma). 
  \cqed\end{proof}

  Hence, if $|F|$ is large enough, there is at least one arc $(u,v)\in F$
  such that $|S_{uv}\cap U|, |T_{uv} \cap W| > 2k$. 
  Naturally, for every $u' \in S_{uv}\cap U$
  there is a corresponding arc $(u',v') \in F$
  and similarly for $T_{uv}$.  
  Since $C_{uv}$ contains at most $k$ arcs,
  this leaves more than $k$ further arcs of each type,
  and the cut $C_{uv}$ fits the requirements of this lemma.
  This finishes the proof.
\end{proof}

From here, Theorem~\ref{thm:dtw-bound} follows easily.
Recall that $T$ is a proposed well-linked set.
Since $G$ is $k$-cut-minimal, every vertex $v \in T$ is incident to some arc that 
participates in a minimal $st$-cut of size at most $k$. This gives us
a collection $F_T \subseteq A$ of arcs as in Lemma~\ref{lemma:balancedcuts},
with $|F_T|\geq |T|/2$. If $|F_T|>h(k)$, then the cut $C$ provided by the lemma
is a witness that $T$ is not well-linked. Hence $|T|$ is bounded by $2^{\Oh(k^2)}$,
and by Theorem~\ref{thm:well-linked} the directed treewidth of $G$ is bounded by $f(k)$.

\section{Conclusions}\label{sec:conc}
We have presented reductions showing $W[1]$- and nearly tight ETH-hardness
of \dirmc{} and \stor{}.

We would like to conclude with two open problems.
The first one is the most natural one: is \dirmc{}, parameterized by the size of the cutset,
fixed-parameter tractable for \emph{three} terminal pairs?
It seems that our main gadget (the $p$-grid that we use to check an auxiliary binary relation) inherently requires
four commodities to work: two that pass through it horizontally and vertically, and two additional for synchronizing the cutpoints
of the input and output bidirectional paths. On the other hand, we were not able to extract any combinatorial property of the
solutions in the three terminal case that would warrant an FPT algorithm.
As a related question, we would also like to ask whether the two terminal case is $W\!K[1]$-hard parameterized by the cutset, or whether it admits a so-called polynomial Turing kernel (see~\cite{turingkernels}). It is known not to admit a polynomial kernel under standard complexity-theoretical assumptions~\cite{cutkernel-lbs}, but nothing is known about Turing kernels (even conjecturally).

Secondly, we would like to repeat from~\cite{chainsat} the question of fixed-parameter tractability of the
$\ell$-\textsc{Chain SAT} problem, parameterized the cutset. In this problem (with a fixed integer $\ell$ in the problem
    description) we are given an integer $k$ and a set of $n$ binary variables, unary constraints, and constraints of the form
$(x_1 \Rightarrow x_2) \wedge (x_2 \Rightarrow x_3) \wedge \ldots \wedge (x_{\ell-1} \Rightarrow x_\ell)$;
one asks to delete at most $k$ constraints to get a satisfiable instance. 
It is easy to see that this problem can be cast as a cut problem in directed graphs via the natural implication/arc correspondence.
While it is reasonable to suspect that a strong structural result for $k$-cut-minimal graphs may lead
to an FPT algorithm for this problem, our lower bound methodology seems irrelevant, as it inherently requires different commodities.
%However, we would also like to point out that we don't know whether $\ell$-\textsc{Chain SAT} is FPT even on DAGs.

\paragraph{Acknowledgements.} We would like to acknowledge a number of insightful discussions 
with various people on the graph cut problems in directed graphs, in particular on $k$-cut-minimal graphs;
we especially thank Anudhyan Boral, Marek Cygan, Alina Ene, Tomasz Kociumaka, Stefan Kratsch, Daniel Loksthanov, D\'{a}niel Marx, Micha\l{} Pilipczuk, Saket Saurabh, and Micha\l{} W\l{}odarczyk.

\bibliographystyle{abbrv}
\bibliography{dir-cuts}

\begin{thebibliography}{10}

\bibitem{kreutzer-dimeas}
S.~A. Amiri, L.~Kaiser, S.~Kreutzer, R.~Rabinovich, and S.~Siebertz.
\newblock Graph searching games and width measures for directed graphs.
\newblock In E.~W. Mayr and N.~Ollinger, editors, {\em 32nd International
  Symposium on Theoretical Aspects of Computer Science, {STACS} 2015, March
  4-7, 2015, Garching, Germany}, volume~30 of {\em LIPIcs}, pages 34--47.
  Schloss Dagstuhl - Leibniz-Zentrum fuer Informatik, 2015.

\bibitem{digraphs}
J.~Bang-Jensen and G.~Gutin.
\newblock {\em Digraphs: Theory, Algorithms and Applications}.
\newblock Springer, 2008.

\bibitem{bousquet}
N.~Bousquet, J.~Daligault, and S.~Thomass{\'{e}}.
\newblock Multicut is {FPT}.
\newblock In L.~Fortnow and S.~P. Vadhan, editors, {\em Proceedings of the 43rd
  {ACM} Symposium on Theory of Computing, {STOC} 2011, San Jose, CA, USA, 6-8
  June 2011}, pages 459--468. {ACM}, 2011.

\bibitem{ChenLL09}
J.~Chen, Y.~Liu, and S.~Lu.
\newblock An improved parameterized algorithm for the minimum node multiway cut
  problem.
\newblock {\em Algorithmica}, 55(1):1--13, 2009.

\bibitem{dfvs-alg}
J.~Chen, Y.~Liu, S.~Lu, B.~O'Sullivan, and I.~Razgon.
\newblock A fixed-parameter algorithm for the directed feedback vertex set
  problem.
\newblock {\em J. {ACM}}, 55(5), 2008.

\bibitem{rand-contr}
R.~H. Chitnis, M.~Cygan, M.~Hajiaghayi, M.~Pilipczuk, and M.~Pilipczuk.
\newblock Designing {FPT} algorithms for cut problems using randomized
  contractions.
\newblock In {\em 53rd Annual {IEEE} Symposium on Foundations of Computer
  Science, {FOCS} 2012, New Brunswick, NJ, USA, October 20-23, 2012}, pages
  460--469. {IEEE} Computer Society, 2012.

\bibitem{chainsat}
R.~H. Chitnis, L.~Egri, and D.~Marx.
\newblock List {H}-coloring a graph by removing few vertices.
\newblock In H.~L. Bodlaender and G.~F. Italiano, editors, {\em Algorithms -
  {ESA} 2013 - 21st Annual European Symposium, Sophia Antipolis, France,
  September 2-4, 2013. Proceedings}, volume 8125 of {\em Lecture Notes in
  Computer Science}, pages 313--324. Springer, 2013.

\bibitem{dmwc-alg}
R.~H. Chitnis, M.~Hajiaghayi, and D.~Marx.
\newblock Fixed-parameter tractability of directed multiway cut parameterized
  by the size of the cutset.
\newblock {\em {SIAM} J. Comput.}, 42(4):1674--1696, 2013.

\bibitem{stor}
M.~Cygan, G.~Kortsarz, and Z.~Nutov.
\newblock Steiner forest orientation problems.
\newblock {\em {SIAM} J. Discrete Math.}, 27(3):1503--1513, 2013.

\bibitem{cutkernel-lbs}
M.~Cygan, S.~Kratsch, M.~Pilipczuk, M.~Pilipczuk, and M.~Wahlstr{\"{o}}m.
\newblock Clique cover and graph separation: New incompressibility results.
\newblock {\em {TOCT}}, 6(2):6, 2014.

\bibitem{bisection-alg}
M.~Cygan, D.~Lokshtanov, M.~Pilipczuk, M.~Pilipczuk, and S.~Saurabh.
\newblock Minimum bisection is fixed parameter tractable.
\newblock In D.~B. Shmoys, editor, {\em Symposium on Theory of Computing,
  {STOC} 2014, New York, NY, USA, May 31 - June 03, 2014}, pages 323--332.
  {ACM}, 2014.

\bibitem{mwc-lp}
M.~Cygan, M.~Pilipczuk, M.~Pilipczuk, and J.~O. Wojtaszczyk.
\newblock On multiway cut parameterized above lower bounds.
\newblock {\em {TOCT}}, 5(1):3, 2013.

\bibitem{sunflower-original}
P.~Erd\H{o}s and R.~Rado.
\newblock Intersection theorems for systems of sets.
\newblock {\em Journal of the London Mathematical Society}, s1-35(1):85--90,
  1960.

\bibitem{flum-grohe}
J.~Flum and M.~Grohe.
\newblock {\em Parameterized Complexity Theory}.
\newblock Springer, 2006.

\bibitem{GanianHKMORS10}
R.~Ganian, P.~Hlinen{\'{y}}, J.~Kneis, D.~Meister, J.~Obdrz{\'{a}}lek,
  P.~Rossmanith, and S.~Sikdar.
\newblock Are there any good digraph width measures?
\newblock In V.~Raman and S.~Saurabh, editors, {\em Parameterized and Exact
  Computation - 5th International Symposium, {IPEC} 2010, Chennai, India,
  December 13-15, 2010. Proceedings}, volume 6478 of {\em Lecture Notes in
  Computer Science}, pages 135--146. Springer, 2010.

\bibitem{GHKLOR-diwidth}
R.~Ganian, P.~Hliněný, J.~Kneis, A.~Langer, J.~Obdržálek, and
  P.~Rossmanith.
\newblock Digraph width measures in parameterized algorithmics.
\newblock {\em Discrete Applied Mathematics}, 168(0):88 -- 107, 2014.
\newblock Fifth Workshop on Graph Classes, Optimization, and Width Parameters,
  Daejeon, Korea, October 2011.

\bibitem{turingkernels}
D.~Hermelin, S.~Kratsch, K.~Soltys, M.~Wahlstr{\"{o}}m, and X.~Wu.
\newblock A completeness theory for polynomial (turing) kernelization.
\newblock {\em Algorithmica}, 71(3):702--730, 2015.

\bibitem{eth}
R.~Impagliazzo, R.~Paturi, and F.~Zane.
\newblock Which problems have strongly exponential complexity?
\newblock {\em J. Comput. Syst. Sci.}, 63(4):512--530, 2001.

\bibitem{dags-alg}
S.~Kratsch, M.~Pilipczuk, M.~Pilipczuk, and M.~Wahlstr{\"{o}}m.
\newblock Fixed-parameter tractability of multicut in directed acyclic graphs.
\newblock {\em {SIAM} J. Discrete Math.}, 29(1):122--144, 2015.

\bibitem{kreutzer-width-chapter}
S.~Kreutzer and S.~Ordyniak.
\newblock Width-measures for directed graphs and algorithmic applications.
\newblock In M.~Dehmer and F.~Emmert-Streib, editors, {\em Quantitative Graph
  Theory: Mathematical Foundations and Applications}. Chapman and Hall/CRC
  Press, 2014.

\bibitem{bulletin}
D.~Lokshtanov, D.~Marx, and S.~Saurabh.
\newblock Lower bounds based on the exponential time hypothesis.
\newblock {\em Bulletin of the {EATCS}}, 105:41--72, 2011.

\bibitem{vc-lp}
D.~Lokshtanov, N.~S. Narayanaswamy, V.~Raman, M.~S. Ramanujan, and S.~Saurabh.
\newblock Faster parameterized algorithms using linear programming.
\newblock {\em {ACM} Transactions on Algorithms}, 11(2):15:1--15:31, 2014.

\bibitem{marx-impsep}
D.~Marx.
\newblock Parameterized graph separation problems.
\newblock {\em Theor. Comput. Sci.}, 351(3):394--406, 2006.

\bibitem{subiso-lb}
D.~Marx.
\newblock Can you beat treewidth?
\newblock {\em Theory of Computing}, 6(1):85--112, 2010.

\bibitem{anti-isolation}
D.~Marx.
\newblock Important separators and parameterized algorithms, 2011.
\newblock Lecture slides; available at
  {\url{http://www.cs.bme.hu/~dmarx/papers/marx-mds-separators-slides.pdf}}.

\bibitem{festschrift}
D.~Marx.
\newblock What's next? {F}uture directions in parameterized complexity.
\newblock In H.~L. Bodlaender, R.~Downey, F.~V. Fomin, and D.~Marx, editors,
  {\em The Multivariate Algorithmic Revolution and Beyond - Essays Dedicated to
  Michael R. Fellows on the Occasion of His 60th Birthday}, volume 7370 of {\em
  Lecture Notes in Computer Science}, pages 469--496. Springer, 2012.

\bibitem{cuts-treewidth}
D.~Marx, B.~O'Sullivan, and I.~Razgon.
\newblock Finding small separators in linear time via treewidth reduction.
\newblock {\em {ACM} Transactions on Algorithms}, 9(4):30, 2013.

\bibitem{marx-razgon}
D.~Marx and I.~Razgon.
\newblock Fixed-parameter tractability of multicut parameterized by the size of
  the cutset.
\newblock {\em {SIAM} J. Comput.}, 43(2):355--388, 2014.

\bibitem{soda}
M.~Pilipczuk and M.~Wahlstr{\"{o}}m.
\newblock Directed multicut is \emph{W}[1]-hard, even for four terminal pairs.
\newblock In R.~Krauthgamer, editor, {\em Proceedings of the Twenty-Seventh
  Annual {ACM-SIAM} Symposium on Discrete Algorithms, {SODA} 2016, Arlington,
  VA, USA, January 10-12, 2016}, pages 1167--1178. {SIAM}, 2016.

\bibitem{a2sat-alg}
I.~Razgon and B.~O'Sullivan.
\newblock Almost 2-{SAT} is fixed-parameter tractable.
\newblock {\em J. Comput. Syst. Sci.}, 75(8):435--450, 2009.

\bibitem{magnus-lp}
M.~Wahlstr{\"{o}}m.
\newblock Half-integrality, {LP}-branching and {FPT} algorithms.
\newblock In C.~Chekuri, editor, {\em Proceedings of the Twenty-Fifth Annual
  {ACM-SIAM} Symposium on Discrete Algorithms, {SODA} 2014, Portland, Oregon,
  USA, January 5-7, 2014}, pages 1762--1781. {SIAM}, 2014.

\end{thebibliography}

\end{document}